\documentclass[twocolumn]{autart}    %

\usepackage{graphicx}          %

\usepackage{amsmath,amssymb,amsfonts}
\usepackage{bm, xcolor, hyperref}

\newcommand{\pExp}[2]{\ensuremath{\mathtt{pExp}^\pi_{\mathcal{F}}(#1 #2)}}
\newcommand{\PExp}[1]{\ensuremath{\widehat{\mathtt{pExp}}^\pi_{\mathcal{F}}(#1)}}
\newcommand{\Exp}[1]{\ensuremath{\mathtt{Exp}^{\pi}_{\mathcal{F}}(#1)}}
\newcommand{\Fix}[1]{\ensuremath{\mathtt{Fix}_{#1}}}

\newcommand{\grad}[0]{\nabla \!}

\DeclareMathOperator*{\st}{subject~to}

\DeclareMathOperator*{\argmin}{arg\,min}

\newtheorem{definition}{Definition}

\newtheorem{proof}{Proof}

\newtheorem{remark}{Remark}
\newtheorem{theorem}{Theorem}
\newtheorem{corollary}{Corollary}

\begin{document}

\begin{frontmatter}
\runtitle{A complete characterization of linearly convergent algorithms}  

\title{Learning to Optimize with Guarantees: A Complete Characterization of Linearly Convergent Algorithms\thanksref{footnoteinfo}} 

\thanks[footnoteinfo]{This paper was not presented at any IFAC 
meeting. Corresponding author: A.~Martin. This work was supported by Digital Futures and the Swiss National Science Foundation (SNSF) through the Ambizione grant PZ00P2\_208951.}

\author[a]{Andrea Martin}\ead{andrmar@kth.se},
\author[b]{Ian R. Manchester}\ead{ian.manchester@sydney.edu.au},
\author[c]{Luca Furieri}\ead{luca.furieri@eng.ox.ac.uk}

\address[a]{School of Electrical Engineering and Computer Science, and Digital Futures, KTH Royal Institute of Technology, Sweden}
\address[b]{Australian Centre for Robotics and School of Aerospace, Mechanical and Mechatronic Engineering, The University of Sydney, Australia}
\address[c]{Department of Engineering Sciences, University of Oxford, United Kingdom}


\begin{abstract} %
The design of many classical optimization algorithms is driven by the certification of linear convergence rates over classes of optimization problems. In this paper, we consider the problem of improving the average-case performance of an algorithm over a specific distribution of problem instances. While this task can be tackled by embedding trainable components into the algorithm updates, a key challenge is to preserve worst-case guarantees across the entire problem class. For classes of composite optimization problems, we show that all linearly convergent algorithms can be parametrized in terms of a baseline linearly convergent algorithm, and a set of trainable, exponentially-decaying modifications to its update rule; crucially, this parametrization excludes all—and only—the algorithms that do not converge linearly. Our results apply to improving the average-case performance of classical algorithms such as gradient descent for nonconvex, gradient-dominated functions; Nesterov’s accelerated method for smooth, strongly convex functions; and projected gradient methods for optimization over polyhedral feasible sets. We illustrate how our characterization can be used for learning to optimize with linear convergence and feasibility guarantees. Numerical results showcase benefits over classical optimizers when solving ill-conditioned systems of linear equations and running a model predictive control scheme on a linear dynamical system.
\end{abstract}

\end{frontmatter}

\section{Introduction}
Guarantees of fast convergence are crucial whenever optimization must be executed under tight computational budgets, as in large-scale machine learning (ML) or real-time model predictive control (MPC). Worst-case linear convergence guarantees have been developed for iterative optimization algorithms over several classes of objective functions, whose structure—e.g., strong convexity, smoothness, or gradient dominance—can be exploited by first-order schemes such as standard gradient descent and Nesterov’s accelerated method \cite{nesterov1983method}. A growing body of work leverages the analogy between worst-case convergence rates and robust-control techniques such as integral quadratic constraints (IQCs), leading to characterizations of accelerated algorithms with provably optimal rates across families of convex functions \cite{lessard2016analysis,scherer2021convex,van2017fastest}.

Worst-case rate guarantees establish a baseline performance in terms of the number of iterations required to achieve a certain level of precision. However, the performance of an algorithm in a specific application does not depend solely on its worst-case convergence rate. First, there exist fundamental trade-offs between the speed of convergence and the robustness of algorithmic behavior; see, for instance, the speed/covariance trade-off for accelerated methods in strongly convex optimization analyzed in \cite{mohammadi2020robustness}. This raises the question of how to appropriately define algorithm performance. A second challenge is that scenarios encountered in applications rarely span the entirety of the space of problem instances for which the worst-case guarantee is tight, resulting in overly conservative average-case performance. This introduces another trade-off: how to tailor an algorithm performance to specific instances without compromising the original uniform guarantees. A prime example of such a situation is MPC \cite{mayne2000nonlinear}, where the optimization problems to be solved online often share the same objective and system dynamics constraints, and only differ in the initial state. In such cases, a solver tailored to this sub-family of problems could converge in significantly fewer iterations than a general-purpose algorithm tuned for optimal worst-case convergence rates.

The learning to optimize (L2O) literature addresses the challenge of adopting user-defined performance metrics beyond mere convergence rates and designs algorithms that are tailored to such metrics using ML. For instance, \cite{andrychowicz2016learning} proposes an algorithm performance metric that balances convergence speed with solution precision, and accordingly designs neural network (NN) update rules. However, general-purpose NN update rules come with no guarantees. Convergence with learned updates has been addressed through possibly conservative safeguarding mechanisms \cite{heaton2023safeguarded}, or by limiting the use of ML components to learning initializations \cite{sambharya2024learning} or to hyperparameter tuning \cite{ichnowski2021accelerating} of classical algorithms.
These approaches demonstrate performance exceeding that of state-of-the-art classical algorithms upon training and inherit their convergence guarantees—yet the reliance on classical update rules inherently limits the expressivity of these learning-based design methods.%

Beyond the optimal tuning of classical algorithms, another line of research seeks to use ML to design entirely new convergent update rules, aiming to discover application-specific shortcuts unknown to classical update rules. This has been achieved by taking simple gradient descent as a baseline and enhancing it through learned optimal deviations from such gradient-based updates. The work \cite{martin2024learning} characterizes the class of all and only those deviation functions that ensure convergence to stationary points in nonconvex, unconstrained smooth optimization, enabling learned optimization for user-defined performance metrics and outperforming finely tuned Adam \cite{kingma2014adam} in NN training. The work \cite{banert2024accelerated} uses deep learning to train deviations from gradient descent and saturates these updates with the norm of measured gradients, ensuring convergence for composite convex optimization. The numerical studies of \cite{martin2024learning,banert2024accelerated} empirically demonstrate that convergence rates superior to those of classical algorithms can be achieved through training on gradient descent perturbations. %

Our main goal is to address a question that has remained open in the literature on L2O. Given any algorithm with guaranteed linear convergence over a class of optimization problems, how can we improve its average-case performance over a subclass of problems of interest without sacrificing its guarantees over the entire class? A theoretical study of these trade-offs is important towards making learned optimization a standard and reliable component of linearly convergent algorithm design.

The goal of characterizing all convergent optimization algorithms for a class of problems is analogous to the canonical control-theoretic problem of parameterizing all stabilizing controllers for a given system, initially studied for linear systems \cite{youla1976modern,anderson1998youla} and more recently for nonlinear and in particular neural-network systems (see, e.g. \cite{van2000l2,galimberti2025parametrizations,manchester2026neural}). Structurally, our approach is analogous in that it augments an existing convergent algorithm with a decaying perturbation, whereas the Youla parameterization augments an existing stabilizing controller with a stable perturbation. Connection can also be drawn between our work and the recently-studied input-to-state stability properties of gradient algorithms applied to optimization problems satisfying the Polyak–Łojasiewicz (PL) condition \cite{cui2025perturbed}. In this paper we study a wider class of optimization algorithms and provide a characterization of perturbations achieving precise bounds on convergence rates.

\emph{Contributions:} Given any existing optimization algorithm that achieves linear convergence to a set of fixed points at a specified rate—henceforth the \emph{baseline algorithm}—our main contributions are as follows. First, we characterize conditions on the baseline algorithm under which adding exponentially decaying perturbations preserves the same linear convergence rate, up to a higher-order polynomial term. These conditions identify fundamental trade-offs between the frequency of perturbations and their worst-case impact on the linear convergence rate. Second, we establish a completeness result for linearly convergent optimization: every update rule that converges linearly at a given rate can be written as the sum of the baseline algorithm and a suitably designed exponentially decaying augmentation of its update rule that preserves linear convergence. Finally, we instantiate our results for augmenting state-of-the-art linearly convergent algorithms, notably 1) Nesterov's accelerated gradient (NAG) \cite{nesterov1983method} method for strongly convex and smooth optimization, 2) gradient descent for classes of nonconvex Polyak–Łojasiewicz (PL) functions, and 3) projected gradient descent for convex optimization with polyhedral feasible sets.
Numerical examples showcase the potential of learned linearly convergent optimization in augmenting the performance of classical algorithms under tight iteration budgets. 

\emph{Notation:} The set of all sequences $\mathbf{x} = (x_0,x_1,x_2,\ldots)$ where $x_t \in \mathbb{R}^n$ for all $t\in \mathbb{N}$ is denoted as $\ell^n$.  %
For a function $g:\mathbb{R}^n \rightarrow \mathbb{R}^m$, we write $g\left(\mathbf{x}\right) = (g(x_0),g(x_1),\ldots) \in \ell^m$. For $m \in \mathbb{N}$, we use $\mathcal{P}_m(x)$ to denote the set of positive and monotonically non-decreasing polynomials of degree at most $m$ in the variable $x$ with $p(1) \geq 1$. We define the set of fixed points of an operator $\pi$, assumed non-empty, as $\Fix{\pi}$. For $m \in \mathbb{N}$ and $\gamma \in (0,1)$, we denote by $\ell_{exp}(m, \gamma)$ the class of signals $\mathbf{x}$ for which there exists a polynomial $p_m(t) \in \mathcal{P}_m(t)$ such that $|x_t| \leq p_m(t) \gamma^t$, where $|\cdot|$ denotes any vector norm.  We write $\mathbb{I}_S$ for the indicator function of a set $S$ ($0$ if $x\in S$, $+\infty$ otherwise). The distance from $x$ to a set $S$ is defined as $\mathrm{dist}(x,S)=\inf_{y \in S}|x-y|$. For $u,v\in\mathbb{R}^n$, $u\le v$ indicates the element‐wise inequality.%

\section{Problem formulation}
We consider composite optimization problems of the form
\begin{align}
    &~\min_{x \in \mathbb{R}^d} \quad f(x) + g(x)\,,\label{eq:composite_opt_prob}
\end{align}
where $x \in \mathbb{R}^d$ is the decision variable, $f : \mathbb{R}^d \to \mathbb{R}$ is \textcolor{black}{proper and }$\beta$-smooth, and $g : \mathbb{R}^d \to \mathbb{R} \cup \{+ \infty\}$ is convex, proper, and lower semi-continuous, but potentially nonsmooth. We let $F(x) = f(x)+g(x)$ for brevity, and we assume that the set of optimizers $\mathcal{X}^\star=\argmin_{x \in \mathbb{R}^d} ~ F(x)$ is non-empty. In particular, we note that \eqref{eq:composite_opt_prob} subsumes constrained optimization problems of the form
\begin{subequations}
\label{eq:constrained_opt_prob}
\begin{align}
    &~\min_{x \in \mathbb{R}^d} \quad f_0(x) \label{eq:constrained_opt_prob_objective}\\
    &\st \quad f_i(x) \leq 0\,, \quad \forall i \in [1,M]\,,\label{eq:constraints_of_opt_problem}
\end{align}
\end{subequations}
where $f_0 : \mathbb{R}^d \to \mathbb{R}$ is $\beta$-smooth,  and each function $f_i : \mathbb{R}^d \to \mathbb{R}$ with $i \in [1,M]$ defines a non-empty convex feasibility set $\mathcal{X}_i \subseteq \mathbb{R}^d$. In fact, one can rewrite \eqref{eq:constrained_opt_prob} as an instance of \eqref{eq:composite_opt_prob} letting $f(x)=f_0(x)$ and $g(x) = \max_{i \in [1,M]}~\mathbb{I}_{\mathcal{X}_i}(x)$.

A standard method to solve problem \eqref{eq:composite_opt_prob} is to analytically construct iterations of the form:
\begin{equation}
    \label{eq:optimization_algorithms_fixed_point}
    \xi_{t+1} = \pi(F, \xi_t)\,,  ~ x_t = \phi(F, \xi_t)\,, ~ \xi_0 \in \mathbb{R}^n\,, ~ t \in \mathbb{N}\,,
\end{equation}
where $\xi_t \in \mathbb{R}^n$ is the state variable, the decision $x_t \in \mathbb{R}^d$ is the output variable, and the operator $\pi$ is designed so that its set of fixed points $\Fix{\pi}$—the set of points $\xi^\star$ such that $\pi(F, \xi^\star)=\xi^\star$—is related to $\mathcal{X}^\star$ through $\phi$; that is, a point $x^\star \in \mathcal{X}^\star$ can be reconstructed from a point $\xi^\star \in \Fix{\pi}$ as per $x^\star = \phi(F, \xi^\star)$.

A key metric for the performance of algorithms \eqref{eq:optimization_algorithms_fixed_point} when applied to a class of problems $F \in \mathcal{F}$ is how fast they converge to $\Fix{\pi}$. Classical optimization algorithms often come with convergence guarantees that hold for the worst-case instance of $F\in \mathcal{F}$. However, optimal control methods such as MPC require efficiently finding solutions to the instances of \eqref{eq:composite_opt_prob} that are encountered during deployment,  
where the objective $F$ is drawn from a specific distribution $\mathcal{D}_{\mathcal{F}}$ over the class $\mathcal{F}$. Motivated as such, in this work we investigate the following question.

\emph{Given a set of problem instances $F \sim \mathcal{D}_{\mathcal{F}}$ and a baseline algorithm $\pi$ to solve \eqref{eq:composite_opt_prob}, how can we improve its average-case performance over $\mathcal{D}_{\mathcal{F}}$, while retaining worst-case convergence guarantees over the entire class $\mathcal{F}$?}

\smallskip

In particular, this paper characterizes algorithms $\bm{\nu}$ that achieve \emph{linear convergence} to $\Fix{\pi}$ for classes of functions $\mathcal{F}$.  We will showcase how to leverage such characterization for learning-based algorithm design in Section~\ref{sec:numerical}. 

\begin{definition}
\label{def:linearly_convergent}
    An algorithm $\xi_{t+1} = \nu_t(F,\xi_{t:0})$ is said to be \textit{linearly convergent} to $\Fix{\pi}$ for $\mathcal{F}$ with rate $\gamma \in (0,1)$ if there exists a polynomial $p(t) \in \mathcal{P}_m(t)$ such that for any $F \in \mathcal{F}$
    \begin{equation}
    \label{eq:exponential_convergence}
        \operatorname{dist}(\xi_{t}, \Fix{\pi}) \leq  p(t) \gamma^t \operatorname{dist}(\xi_0, \Fix{\pi})\,, ~ \forall \xi_0 \in \mathbb{R}^n\,,
    \end{equation}
    at all times, where $\operatorname{dist}(\cdot, \cdot)$ is a distance function.  In this case, we write that $\bm{\nu} \in \pExp{m, \gamma}$. When the focus is not on the polynomial order and only on the exponential convergence rate, we write $\bm{\nu} \in \PExp{\gamma}$. Additionally, if \eqref{eq:exponential_convergence} holds with a  polynomial such that $p(1) = 1$, then we say that $\bm{\nu}$ is monotonically linearly convergent to $\Fix{\pi}$ since the distance to the set of fixed points shrinks monotonically with the iterations, and we write $\bm{\nu} \in \Exp{\gamma}$.\footnote{Note that if \eqref{eq:exponential_convergence} holds with a  polynomial $p(t)$ such that $p(1) = 1$, then $\operatorname{dist}(\xi_{t}, \Fix{\pi}) \leq  \gamma \operatorname{dist}(\xi_{t-1}, \Fix{\pi}) \leq  \dots \leq \gamma^t \operatorname{dist}(\xi_{0}, \Fix{\pi})$. Hence, \eqref{eq:exponential_convergence} holds with the constant polynomial $p(t) = 1$ and $\operatorname{dist}(\xi_{t}, \Fix{\pi})$ shrinks monotonically.}
\end{definition}

To enhance the average-case performance of a baseline algorithm $\pi$ on problem instances $F \sim \mathcal{D}_\mathcal{F}$, we aim to design augmentations $v_t \in \mathbb{R}^n$ to its update rule that do not jeopardize its convergence guarantees. Specifically, we will consider augmented update rules defined by the iterations:
\begin{align}
    \label{eq:optimization_algorithms_fixed_point_plus_v}
    \xi_{t+1} &= \nu_t(F, \xi_{t:0})= \pi(F, \xi_t) + v_t(F, \xi_{t:0}), ~ \xi_0 \in \mathbb{R}^n,
\end{align}
and we will establish conditions ensuring that $\bm{\nu}$ is linearly convergent as per \eqref{eq:exponential_convergence}. 

\section{A complete characterization of linearly convergent optimization algorithms}
\label{sec:main_results}
In this section, we establish our main results on how introducing an augmentation term $v_t$ affects the worst-case linear convergence guarantees of a given baseline optimization algorithm $\pi$. We first abstract away from the specific form of $\pi$ and the class of functions $\mathcal{F}$ it is designed to optimize; we only assume that $\pi$ is a linearly convergent fixed-point algorithm as per Definition~\ref{def:linearly_convergent}. In Section~\ref{subsec:application}, we present corollaries that reveal several classes of problems \eqref{eq:composite_opt_prob} and corresponding baseline algorithms $\pi$ that are compatible with our framework. 

Consider a baseline algorithm $\pi \in \PExp{\gamma}$ achieving linear convergence as per Definition~\ref{def:linearly_convergent}. The property \eqref{eq:exponential_convergence} implies that the signal $\operatorname{dist}(\xi_t,\Fix{\pi})$ decays exponentially up to a polynomial factor $p(t)$, for any initial condition $\xi_0 \in \mathbb{R}^n$ and any objective $F \in \mathcal{F}$. We now characterize to what extent injecting exponentially decaying signals $v_t$ in the iterates of \eqref{eq:optimization_algorithms_fixed_point_plus_v} can deteriorate the convergence guarantee of $\pi$.

\begin{theorem}
    \label{th:sufficiency_polyexp_general}
    Consider the recursion \eqref{eq:optimization_algorithms_fixed_point_plus_v} and assume that $\pi\in \pExp{m,\gamma}$. Choose any $N \in \mathbb{N}$ such that 
    \begin{equation*}
        \rho=p(N) \gamma^N < 1\,,
    \end{equation*}
    and any auxiliary signal $\mathbf{w}\in \ell_{exp}(m, \rho)$. For every $t \in \mathbb{N}$, construct the augmentation signal $v_t$ in \eqref{eq:optimization_algorithms_fixed_point_plus_v} as follows:
    \begin{equation}
        \label{eq:v_construction_zeros}
        v_t= \begin{cases}w_{\frac{t+1}{N} - 1} \quad \text{ if }t+1 \operatorname{mod} N = 0\,,\\
        0 \qquad\qquad \text{otherwise}\,.
        \end{cases}
    \end{equation}
Then, the iterates of \eqref{eq:optimization_algorithms_fixed_point_plus_v} satisfy:
    \begin{equation}
    \label{eq:new_rate_general}
        \operatorname{dist}(\bm{\xi}, \Fix{\pi}) \in \ell_{exp}(m+1, \sqrt[N]{p(N)}\gamma)\,.
    \end{equation}
\end{theorem}

We report the proof of Theorem~\ref{th:sufficiency_polyexp_general} in Appendix~\ref{app:proof_sufficiency_polyexp_general}. Theorem~\ref{th:sufficiency_polyexp_general} establishes a trade-off between how often we inject an exponentially decaying perturbation—as measured by $N \in \mathbb{N}$—and the degradation of the convergence rate. In particular, if $\pi \in \Exp{\gamma}$, then $p(1) = 1$ by definition and we observe that the asymptotic rate $\gamma$ does not change even with $N=1$, as only the order of the polynomial factor in \eqref{eq:exponential_convergence} is affected. This means that one can apply perturbation $\mathbf{v} \in \ell_{exp}(m, \gamma)$ at every time instant as per \eqref{eq:optimization_algorithms_fixed_point_plus_v} without deteriorating the asymptotic convergence rate. For the general case where $\pi \in \pExp{m, \gamma}$, instead, the convergence rate increases at most to the value $\sqrt[N]{p(N)}\gamma \in (\gamma, 1)$. As expected, when $N$ tends to infinity, we recover the original rate of the baseline algorithm because $\lim_{N \to \infty} \sqrt[N]{p(N)} = 1$; this corresponds to the case $v_t= 0$ for all $t$.

It remains crucial to understand how large is the class of linearly convergent algorithms $\xi_{t} = \nu_t(F, \xi_{t:0})$ that can be achieved by perturbing a baseline algorithm $\pi$ with an exponentially decaying learned update $v_t$ as per \eqref{eq:optimization_algorithms_fixed_point_plus_v}. Our next result establishes conditions on the baseline algorithm $\pi$ that guarantee that \emph{any} algorithm in $\pExp{m,\gamma}$ can be represented—provided that such target algorithm satisfies the following regularity condition.

\begin{definition}
\label{def:regularity} Given a linearly convergent algorithm $\bm{\nu} \in \pExp{m,\gamma}$ with iterates $\bm{\xi}$, define the sequence of updates $\mathbf{u}$ satisfying \[u_t=\xi_{t+1}-\xi_t.\] 
We say that $\bm{\nu}$ is \emph{regular} if this sequence of updates vanishes with the same exponential rate, i.e. $
        \mathbf{u} \in \ell_{exp}(m, \gamma)$
\end{definition}
In other words, the definition above excludes pathological cases of linearly convergent algorithms that can cycle indefinitely among the points in $\Fix{\pi}$ even when $\operatorname{dist}(\xi_{t},\Fix{\pi})=0$. As an example, consider the case where the objective $F$ is constant, the baseline algorithm is $\xi_{t+1} = \pi(F, \xi_t) = \xi_t$ with $\xi_0 \neq 0$, and $\Fix{\pi} = \mathbb{R}^n$. In this case, the algorithm $\xi_{t+1} = \nu_t(F, \xi_{t:0}) = -\xi_t$ converges to $\Fix{\pi}$ yet is not regular according to Definition~\ref{def:regularity}. In fact, by definition of $\bm{\nu}$, we have $u_t = -2\xi_0 \neq 0$ when $t$ is even and $u_t = 2\xi_0 \neq 0$ when $t$ is odd, hence the resulting $\mathbf{u}$ is not exponentially decaying. 

We are ready to present our completeness result.

\begin{theorem}
\label{th:necessity}
Let $\pi\in \Exp{\gamma}$ be such that $\pi(F,\xi)$ is Lipschitz continuous in $\xi$. Consider the augmented algorithm $\bm{\nu}$ with  iterates $\xi_t$ defined as per \eqref{eq:optimization_algorithms_fixed_point_plus_v}, and any  target algorithm $\chi_{t+1} = \sigma_t(F, \chi_{t:0})$ such that  $\bm{\sigma} \in \pExp{m,\gamma}$. If $\bm{\sigma}$ is regular, there exists a sequence $\mathbf{v}(F,\bm{\xi}) \in \ell_{exp}(m, \gamma)$ such that the iterations of $\bm{\nu}$ initialized with $\xi_0 = \chi_0$ are equivalent to those of $\bm{\sigma}$. Additionally, the augmented algorithm $\bm{\nu}$ belongs to $\PExp{\gamma}$ for any $\mathbf{v} \in \ell_{exp}(m,\gamma)$ with $m \in \mathbb{N}$.
\end{theorem} 

A few remarks are in order. First, the completeness property of Theorem~\ref{th:necessity} is key in the context of automating the design of augmented algorithms, e.g., by learning from sampled problems, as it implies that \eqref{eq:optimization_algorithms_fixed_point_plus_v} encompasses all regular linearly convergent algorithms  $\bm{\sigma} \in \pExp{m,\gamma}$. Second, when we learn an augmentation term $\mathbf{v}\in\ell_{exp}(m,\gamma)$ by searching over the entire space of exponentially decaying updates, it is desirable that the baseline algorithm satisfies the stronger condition $\pi\in\Exp{\gamma}$. This is because, if $\pi\in\pExp{m,\gamma} ~\setminus \Exp{\gamma}$, Theorem~\ref{th:sufficiency_polyexp_general} only guarantees linear convergence—with the worsened rate $\sqrt[N]{p(N)}\,\gamma$—for those $\mathbf{v}$ chosen exactly as in \eqref{eq:v_construction_zeros}.  In other words, without $\pi\in\Exp{\gamma}$, most perturbations in $\ell_{exp}(m,\gamma)$ may not preserve linear convergence. Third, the assumption that $\pi(F,\xi)$ is Lipschitz continuous in $\xi$ is mild; we will show in the following that this condition is satisfied for baseline algorithms widely used for convex and composite optimization.

The rest of this section is dedicated to establishing how Theorem~\ref{th:sufficiency_polyexp_general} and Theorem~\ref{th:necessity} can be used to augment existing solvers for composite optimization problems in the form \eqref{eq:composite_opt_prob} drawn from specific classes $\mathcal{F}$.

\subsection{Results for smooth optimization}
\label{subsec:application}
 We first consider the case \eqref{eq:composite_opt_prob} where $g(x)=0$ for all $x \in \mathbb{R}^d$, leaving us with the task of minimizing a $\beta$-smooth function $F(x)=f(x)$. Our first result focuses on classes of possibly nonconvex functions for which standard gradient descent achieves monotonic linear convergence.

\begin{corollary}
\label{co:RSI}
    Let $\mathcal{F}_{RSI}^{\beta,\mu}$ be the class of $\beta$-smooth functions satisfying the restricted secant inequality (RSI) with constant $\mu > 0$, that is, those for which it holds
    \begin{equation}
    \label{eq:RSI_condition}
        \nabla F(x)^\top (x-x^{\star}) \geq \frac{\mu}{2} \operatorname{dist}(x, \mathcal{X}^\star)^2\,, ~ \forall x \in \mathbb{R}^d\,,
    \end{equation}
    for any $x^\star \in \argmin_{\mathcal{X}^\star} \operatorname{dist}(x,\mathcal{X}^\star)$. Let $\pi$ be the gradient descent update rule $\pi(F,\xi_t)=\xi_t-\eta\nabla F(\xi_t)$ with $\xi_t = x_t$, $\eta=\frac{\mu}{\beta^2}$, and $\gamma = \sqrt{1-\frac{\mu^2}{\beta^2}}$. Then, any regular algorithm $\bm{\sigma}\in \mathtt{pExp}_{\mathcal{F}_{RSI}^{\beta,\mu}}^\pi(m,\gamma)$ can be written as 
    \begin{equation}
        \label{eq:augmented_GD_RSI}
        x_{t+1} = \nu_t(F, x_{t:0}) = x_t-\eta \nabla F(x_t)+v_t(F,x_{t:0})\,,
    \end{equation}
    with $\mathbf{v} \in \ell_{exp}(m,\gamma)$. Vice versa, for any $\mathbf{v} \in \ell_{exp}(m,\gamma)$, the algorithm \eqref{eq:augmented_GD_RSI} is such that $\bm{\nu} \in \widehat{\mathtt{pExp}}_{\mathcal{F}_{RSI}^{\beta,\mu}}^\pi(\gamma)$.
\end{corollary}
\begin{proof}
    By Theorem 2.1 of \cite{liao2024error}, it holds that \eqref{eq:exponential_convergence} holds for the gradient descent algorithm $\xi_{t+1}=\pi(F, \xi_t)=\xi_t-\eta\nabla F(\xi_t)$ with $\eta = \frac{\mu}{\beta^2}$ with $\gamma = \sqrt{1-\frac{\mu^2}{\beta^2}} \in (0,1)$. Further, we have that $\pi(F, \xi_t)$ is Lipschitz continuous in $\xi_t$ since $|\pi(F, x) - \pi(F,y)| = |x - y - \eta \nabla F(x) + \eta \nabla F(y)| \leq (1+\eta\beta) |x-y|$, where the last inequality follows from the $\beta$-smoothness of $F \in \mathcal{F}_{RSI}^{\beta,\mu}$. The result then follows by applying Theorem~\ref{th:necessity}.
\end{proof}

\smallskip

The result of Corollary~\ref{co:RSI} enables learning over the class of \emph{all} the regular linearly convergent algorithms in $\mathtt{pExp}_{\mathcal{F}_{RSI}^{\beta,\mu}}^\pi(m,\gamma)$, while ensuring that the augmented algorithm \eqref{eq:augmented_GD_RSI} never leaves the class $\widehat{\mathtt{pExp}}_{\mathcal{F}_{RSI}^{\beta,\mu}}^\pi(\gamma)$, irrespectively of how ``badly'' the augmentation term $\mathbf{v}\in \ell_{exp}(m,\gamma)$ may be chosen. 

A few comments regarding the generality of the class of functions $\mathcal{F}_{RSI}^{\beta,\mu}$  are in order. First, $\mathcal{F}_{RSI}$ encompasses certain nonconvex functions, as highlighted in \cite{zhang2013gradient}. Second, it holds that $\mathcal{F}_{SC}^{\beta,\mu}\subset \mathcal{F}_{cPL}^{\beta,\mu}\subset \mathcal{F}_{RSI}^{\beta,\mu}$ , where $\mathcal{F}_{SC}^{\beta,\mu}$ is the set of $\beta$-smooth  and $\mu$-strongly convex functions complying with
\begin{equation}
    \label{eq:SC}
    F(y)\geq F(x)+\nabla F(x)^\top(y-x)+\frac{\mu}{2}|y-x|^2\,,
\end{equation}
for some $\mu>0$, and $\mathcal{F}_{cPL}^{\beta,\mu}$  is the set of all the $\beta$-smooth and convex functions  that  comply with the Polyak--\L{}ojasiewicz (PL) inequality 
\begin{equation}
    \label{eq:PL_condition}
    F(x)-\min_{x \in \mathbb{R}^d} F(x)\leq \frac{1}{2\mu}|\nabla F(x)|^2\,,
\end{equation}
for some $\mu>0$.
\begin{remark}
    It is well known that $\mathcal{F}_{RSI}^{\beta,\mu} \subseteq \mathcal{F}_{PL}^{\beta,\frac{\mu^2}{4\beta}}$, where $\mathcal{F}_{PL}^{\beta,\mu}$ is the set of all possibly nonconvex functions satisfying \eqref{eq:PL_condition}, see \cite{karimi2016linear}. For functions in $\mathcal{F}_{PL}^{\beta,\mu}$, the gradient descent rule $\pi(F,x)=-\frac{1}{\beta} \nabla F(x)$  achieves linear convergence in the function value as per 
    \begin{equation*}
        F(x_t)-F^\star\leq \left(1-\frac{\mu}{\beta}\right)^t(F(x_0)-F^\star)\,,
    \end{equation*}
    where $F^\star = \min_{x \in \mathbb{R}^d} ~F(x)$. However,  $\pi$ induces a monotonically linearly convergent sequence of iterates only if the restricted secant inequality \eqref{eq:RSI_condition} also holds, see \cite{liao2024error}.
\end{remark}

 Corollary~\ref{co:RSI} ensures a complete parametrization of regular linearly convergent algorithms with the same rate $\gamma$ as gradient descent for all functions in $\mathcal{F}_{RSI}^{\beta,\mu}$. For the special case of strongly convex functions $F \in \mathcal{F}_{SC}^{\beta,\mu}$, one typically wants to augment ad-hoc algorithms tailored to $\mathcal{F}_{SC}^{\beta,\mu}$ such as NAG method \cite{nesterov1983method}, the heavy-ball method \cite{polyak1964some,ghadimi2015global}, or optimal-rate algorithms such as those characterized in \cite{lessard2016analysis,van2017fastest}. 

Motivated as such, we show compatibility of the proposed framework with the augmentation of accelerated algorithms for objectives $F \in \mathcal{F}_{SC}^{\beta,\mu}$. %
\begin{corollary}
\label{co:NAG}
Consider the NAG algorithm 
    \begin{equation}
    \label{eq:nesterov_alg}
        \pi(F,\xi_t) = 
        \begin{bmatrix}
            1 + \alpha & -\alpha\\
            1 & 0
        \end{bmatrix} \xi_t + 
        \begin{bmatrix}
            -\eta\\
            0
        \end{bmatrix}
        \grad F\left(\begin{bmatrix}
            1 + \alpha & -\alpha
        \end{bmatrix} \xi_t \right)\,,
    \end{equation}
where  $\xi_t = \begin{bmatrix}
        x_t^\top & x_{t-1}^\top
    \end{bmatrix}^\top$
    and $\alpha \geq 0$ is the momentum coefficient. Let $\eta=\frac{1}{\beta}$ and $\alpha = \frac{\sqrt{\kappa}-1}{\sqrt{\kappa}+1}$, where $\kappa = \frac{\beta}{\mu}\geq 1$ is the condition number.  Choose any target rate degradation factor $\tau \in (1, \frac{1}{\gamma})$, where $\gamma = \sqrt{1 - \frac{1}{\sqrt{\kappa}}}$. Then, for any $N \in \mathbb{N}$ such that $p(N) < \tau^{N}$ and $\mathbf{v}$ constructed as per \eqref{eq:v_construction_zeros} using any $\mathbf{w} \in \ell_{exp}(m, \tau\gamma)$, the augmented algorithm $\bm{\nu}(F, \xi_{t:0})$ defined by $\xi_{t+1} = \pi(F,\xi_t) + v_t\,, $ is such that $\bm{\nu} \in \widehat{\mathtt{pExp}}_{\mathcal{F}_{SC}^{\beta,\mu}}^\pi(\tau \gamma)$.
\end{corollary}
\begin{proof}    
    It is well known that the NAG algorithm \eqref{eq:nesterov_alg} applied to the class $\mathcal{F}_{SC}^{\beta, \mu}$ with the parameters $\alpha$ and $\eta$ as above is such that $\pi \in \pExp{0, \gamma}$, see \cite{nesterov1983method,lessard2016analysis}. Since $\sqrt[N]{p(N)} < \tau < \frac{1}{\gamma}$, we have that $p(N) \gamma^N < 1$ and Theorem~\ref{th:sufficiency_polyexp_general} applies.
\end{proof}

While Corollary~\ref{co:NAG} focuses on the case where  NAG is used as the baseline algorithm $\pi$ in \eqref{eq:optimization_algorithms_fixed_point_plus_v}, we remark that the results extend analogously to any baseline algorithm $\pi \in \mathtt{pExp}_{\mathcal{F}_{SC}^{\beta,\mu}}^\pi(m,\gamma)$ such as those with optimal convergence rates designed using IQCs as per \cite{lessard2016analysis,van2017fastest}. As also discussed after Theorem~\ref{th:sufficiency_polyexp_general}, we note that enhancing accelerated algorithms, which are not monotonic in general, involves a trade-off between keeping the worst-case degradation rate $\tau$ as small as possible and the frequency at which we can apply a learned update. Last, we remark that one can always impose a target $\tau \in (1, \frac{1}{\gamma})$. Indeed, a large enough $N \in \mathbb{N}$ such that $p(N) < \tau^N$ always exists since, for $N$ large enough, the exponential term dominates over the polynomial one.

\subsection{Results for nonsmooth optimization}

We now turn our attention to the case \eqref{eq:composite_opt_prob} where the objective $F(x) = f(x) + g(x)$ is nonsmooth. Our first result focuses on the class $\mathcal{F}_{cPL}^{\!~\infty, \mu}$ of potentially nonsmooth proper, lower semi-continuous, convex functions that comply with the following inequality
\begin{equation}
    \label{eq:PL_condition_nonsmooth}
    F(x)-\min_{x \in \mathbb{R}^d} F(x)\leq \frac{1}{2\mu}\operatorname{dist}(0, \partial F(x))^2\,,
\end{equation}
where $\partial F(x)$ is the convex subdifferential of $F$ at $x$, defined as 
\begin{equation*}
    \partial F(x) = \left\{s \in \mathbb{R}^d : F(y) \geq F(x) + s^\top (y-x)\,, ~\forall y \in \mathbb{R}^d\right\}\,.
\end{equation*}
In particular, note that \eqref{eq:PL_condition_nonsmooth} corresponds to \eqref{eq:PL_condition} when $F$ is differentiable.
\begin{corollary}
\label{co:cPL_nonsmooth}
    Consider the class of functions $F \in \mathcal{F}_{cPL}^{\!~\infty,\mu}$. Let $\pi$ be the proximal point algorithm performing the iterations \begin{equation}
        \label{eq:PPM_original_cPL_nonsmooth}
        x_{t+1} = \operatorname{prox}_F^c(x_t) = \min_{x \in \mathbb{R}^d}~F(x) + \frac{1}{2c} |x - x_t|^2\,, 
    \end{equation}
    where $c > 0$. Let $\gamma = \min\left\{\frac{1}{\sqrt{1 + c \mu}}, \frac{1}{\sqrt{1+\frac{c^2}{\beta \mu}}}\right\} \in (0,1)$. Then, any regular algorithm $\bm{\sigma} \in \mathtt{pExp}_{\mathcal{F}_{cPL}^{\infty,\mu}}^\pi(m,\gamma)$ can be written as 
    \begin{equation}
        \label{eq:augmented_PPM_cPL_nonsmooth}
        x_{t+1} = \nu_t(F, x_{t:0}) = \operatorname{prox}_F^c(x_t) + v_t(F,x_{t:0})\,,
    \end{equation}
    with $\mathbf{v} \in \ell_{exp}(m,\gamma)$. Vice versa, for any $\mathbf{v} \in \ell_{exp}(m,\gamma)$, the algorithm \eqref{eq:augmented_PPM_cPL_nonsmooth} is such that $\bm{\nu} \in \widehat{\mathtt{pExp}}_{\mathcal{F}_{cPL}^{\infty,\mu}}^\pi(\gamma)$.
\end{corollary}
\begin{proof}
    We begin by showing the baseline algorithm $\pi(F,x_t)$ in \eqref{eq:PPM_original_cPL_nonsmooth} is $1$-Lipschitz continuous in $x_t$. Let $u = \pi(F, x)$ and $v = \pi(F, y)$. Since $u$ and $v$ are obtained by solving the minimization \eqref{eq:PPM_original_cPL_nonsmooth}, it holds that $0 \in \partial F(u) + \frac{1}{c}(u-x)$ and $0 \in \partial F(v) + \frac{1}{c}(v-y)$ by the first-order optimality conditions. In other words, there exist subgradients $g_u \in \partial F(u)$ and $g_v \in \partial F(v)$ such that $x - u = c g_u$ and $y - v = c g_u$. Since $F$ is closed, convex, and proper, $\partial F$ is maximally monotone \cite{ryu2016primer} and $(g_u - g_v)^\top (u - v) \geq 0$. Observe that
    \begin{align*}
        (g_u - g_v)^\top (u - v) &= \frac{1}{c} (x-u - (y-v))^\top (u-v)\\
        &= \frac{1}{c} (x-y)^\top (u-v) - \frac{1}{c} |u-v|^2 \geq 0\,.
    \end{align*}
    Therefore, we have $(x-y)^\top (u-v) \geq |u-v|^2$. On the other hand, by the Cauchy-Schwarz inequality, $|x-y||u-v| \geq (x-y)^\top (u-v)$. Combining the two last inequalities, we have $|u-v|^2 \leq |x-y||u-v|$, from which $1$-Lipschitz continuity follows. Similarly to Corollary~\ref{co:RSI}, the result then follows by combining our Theorem~\ref{th:necessity} with the linear convergence result of the proximal point method \eqref{eq:PPM_original_cPL_nonsmooth} when applied to functions $F \in \mathcal{F}_{cPL}^{\infty, \mu}$ from \cite[Theorem~4.2]{liao2024error} and the definitions of error bound and quadratic growth from \cite{karimi2016linear}.
\end{proof}

The result of Corollary~\ref{co:cPL_nonsmooth} holds for any objective $F \in \mathcal{F}_{cPL}^{\infty,\mu}$. In particular, $\mathcal{F}_{cPL}^{\infty,\mu}$ encompasses the class of optimization problems \eqref{eq:composite_opt_prob}, where $f \in \mathcal{F}_{SC}^{\beta, \mu}$ and $g \in \mathcal{F}_{C}^{\infty}$, that is, $g(x)$ is nonsmooth and convex, for which ad-hoc algorithms have been developed to exploit the structure underlying these composite  problems. Our next result focuses on the case where $g(x)$ represents the indicator function of a set of linear   constraints~\eqref{eq:constraints_of_opt_problem} to address constrained optimization problems of the form \eqref{eq:constrained_opt_prob} with $f_0 \in \mathcal{F}_{SC}^{\beta, \mu}$. 

\begin{corollary}
\label{co:constrained}
    Consider the constrained optimization problem~\eqref{eq:constrained_opt_prob} with $f_0 \in \mathcal{F}_{SC}^{\beta, \mu}$ and $f_i(x) = A_ix - b_i$ for all $i \in [1,M]$ and define the feasible set $\mathcal{X} = \{x\in \mathbb{R}^d : f_i(x) \leq 0, ~\forall i \in [1,M]\}$. Let $g(x) = \mathbb{I}_{\mathcal{X}}(x)$ and define $\mathcal{F}_{\text{comp}}$ as the set of all such functions $F(x) = f_0(x) + g(x)$. Let
    $\pi$ be the proximal gradient descent method performing the iterations
    \begin{align}
    \label{eq:base_Proj_composite}
        x_{t+1} &= \min_{x\in \mathbb{R}^d}~g(x) + \frac{1}{2}|x - (x_t - \eta \nabla f(x_t))|^2\\
        &= \operatorname{prox}_{g}(x_t - \eta \nabla f(x_t)) = \operatorname{proj}_{\mathcal{X}}(x_t - \eta \nabla f(x_t))\,,\nonumber
    \end{align}
    where $\eta \in (0, \frac{1}{\beta}]$. Consider any regular algorithm $\chi_{t+1} = \sigma_t(F, \chi_{t:0})$ with feasible iterates $\chi_t \in \mathcal{X}$ and such that $\bm{\sigma}\in \mathtt{pExp}_{\mathcal{F}_{\text{comp}}}^\pi(m,\gamma)$, with $\gamma = 1-\eta\mu$. Then, there exists $\mathbf{v} \in \ell_{exp}(m,\gamma)$ with
    \begin{equation}
    \label{eq:constraints_v}
         A_i v_t \leq b_i - A_i \operatorname{proj}_{\mathcal{X}}(x_t - \eta \nabla f(x_t))\,, ~ \forall i \in [1,M]\,,
    \end{equation}
     at all times such that the augmented algorithm
    \begin{align}
        \label{eq:augmented_Proj_composite}
        x_{t+1} &= \nu_t(F, x_{t:0})\\
        &= \operatorname{proj}_\mathcal{X}(x_t - \eta \nabla f(x_t)) + v_t(F,x_{t:0}) \nonumber\,,
    \end{align}
    is equivalent to $\bm{\sigma}$. Vice versa, for any $\mathbf{v} \in \ell_{exp}(m,\gamma)$ such that \eqref{eq:constraints_v} holds at all times and for all $i \in [1,M]$, the algorithm \eqref{eq:augmented_Proj_composite} is such that the iterates $x_t \in \mathcal{X}$ and $\bm{\nu} \in \widehat{\mathtt{pExp}}_{\mathcal{F}_{\text{comp}}}^\pi(\gamma)$.
\end{corollary}
\begin{proof}
The baseline algorithm $\pi$ defined in \eqref{eq:base_Proj_composite} is such that $\pi\in \mathtt{Exp}_{\mathcal{F}_{\text{comp}}}^\pi(1-\eta\mu)$ as shown in \cite[ Theorem~11.5]{garrigos2023handbook}. Now, consider any regular algorithm $\chi_{t+1} = \sigma_t(F, \chi_{t:0})$ with feasible iterates $\chi_t \in \mathcal{X}$ and such that $\bm{\sigma}\in \mathtt{pExp}_{\mathcal{F}_{\text{comp}}}^\pi(m,\gamma)$, with $\gamma = 1-\eta\mu$. Its iterates are equivalent to those of \eqref{eq:augmented_Proj_composite} by choosing
\begin{equation}
    v_t = -\pi(F,\chi_t)+\sigma_t(F,\chi_{t:0})\,, \quad x_0=\chi_0\,.
\end{equation}
Next, we verify that $A_i v_t \leq b_i - A_i \operatorname{proj}_{\mathcal{X}}(x_t - \eta \nabla f(x_t))$ for all $i \in [1,M]$. We have
\begin{align}
    A_iv_t&=-A_i\pi(F,\chi_t)+A_i\sigma_t(F,\chi_{t:0})\nonumber\\
    &=-A_i\operatorname{proj}_\mathcal{X}\left(\chi_t-\eta \nabla f(\chi_t)\right)+A_i \chi_{t+1}\,.\label{eq:condition_feasibility_vt}
\end{align}
By definition, $A_i \chi_{t+1} \leq b_i$ because $\chi_{t+1}$ lies in $\mathcal{X}$ and the claim follows by direct substitution in \eqref{eq:condition_feasibility_vt} since the sequence $\chi_t$ is equal to $x_t$. Next, we verify that $\pi(F,x)$ is Lipschitz in $x$. Since the projection onto an affine subspace is $1$-Lipschitz, it holds that
\begin{align*}
    |\pi(F,x)-\pi(F,y)|&\leq |x-y+\eta \nabla f(y)-\eta \nabla f(x)|\\
    &\leq (1+\eta \beta)|x-y|\,.
\end{align*}
Last, analogously to the proof of Theorem~\ref{th:necessity}, it holds that $\mathbf{v} \in \ell_{exp}(m,1-\eta \mu)$ because $\bm{\sigma}$ is regular.  Vice versa, if $\mathbf{v} \in \ell_{exp}(m,\gamma)$ is such that at all times $A_i v_t \leq b_i - A_i \operatorname{proj}_{\mathcal{X}}(x_t - \eta \nabla f(x_t))$ for all $i \in [1,M]$, the iterates of \eqref{eq:augmented_Proj_composite} are feasible because 
\begin{align}
    A_i x_{t+1} = A_i  \operatorname{proj}_{\mathcal{X}}(x_t - \eta \nabla f(x_t)) + A_i v_t \leq b_i\,,
\end{align}
and $\bm{\nu}\in \widehat{\mathtt{pExp}}_{\mathcal{F}_{\text{comp}}}^\pi(1-\eta \mu)$ by Theorem~\ref{th:sufficiency_polyexp_general}.
\end{proof}

Leveraging the composite structure of \eqref{eq:constrained_opt_prob}, Corollary~\ref{co:constrained} addresses the requirement of ensuring feasibility of all iterates of \eqref{eq:augmented_Proj_composite} in optimization problems with polytopic constraints. In fact, while Corollary~\ref{co:cPL_nonsmooth} guarantees convergence rates of the augmented algorithm $\bm{\nu}$, feasibility of iterates $x_t$ of \eqref{eq:augmented_PPM_cPL_nonsmooth} may be lost for arbitrary choices of $\mathbf{v} \in \ell_{exp}(m, \gamma)$.  

\section{Learning to optimize with linear convergence guarantees}
\label{sec:numerical}
In this section, we illustrate a natural application of our characterization of linearly convergent algorithms to design new update rules with improved average‑case performance relative to a baseline linearly convergent algorithm while preserving its worst‑case convergence rate. We begin by formulating a meta‑optimization problem over the augmentation sequence $v_t(F, \xi_{t:0})$, where we aim at minimizing the expected cost incurred by the augmented algorithm under a given data distribution while enforcing a target worst‑case convergence rate. Beyond the special case in which the algorithmic costs are quadratic in the iterates, the augmentation sequences are linear functions, and the objective functions $F \in \mathcal{F}$ are quadratic, the meta‑optimization is nonconvex. This motivates parameterizing $v_t(F, \xi_{t:0})$ as an exponentially decaying NN update and learning its parameters via empirical algorithmic cost minimization. We conclude by discussing the significant speed-ups achieved by learned augmentations through examples drawn from linear regression and linear quadratic MPC.\footnote{The Python source code that reproduces our experiments is publicly available at \href{https://github.com/andrea-martin/LinearlyConvergentL2O}{https://github.com/andrea-martin/LinearlyConvergentL2O}.}

\subsection{Augmenting average-case performance of a baseline algorithm through NN updates}
Let $\mathcal{D}_\mathcal{F}$ denote a distribution of objective functions $F \in \mathcal{F}$ and let $\pi \in \pExp{m, \gamma}$ be a known linearly convergent baseline algorithm for $\mathcal{F}$. The problem of designing an augmentation signal $v_t$ to enhance the average-case performance of $\pi$ relative to $\mathcal{D}_\mathcal{F}$ is expressed as
\begin{subequations}
\label{eq:metaopt_constrained_problem}
    \begin{alignat}{3}
   &~\min_{\mathcal{D}_\theta} && \mathbb{E}_{\theta \sim \mathcal{D}_\theta} ~ \mathbb{E}_{\substack{F \sim \mathcal{D}_\mathcal{F}\\ \xi_0 \sim \mathcal{D}_{0}}} \left[ \mathtt{AlgoCost}(F, \bm{\xi}) \right] \label{eq:expected_meta_loss} \\
   &\st && \quad \xi_{t+1} = \xi_t + \pi(F,\xi_{t})+v_t(F, \xi_{t:0}, \theta),\quad \label{eq:algo_def_metaloss_problem} \\
   &~&& \quad \bm{\xi} \in \ell_{exp}(\overline{m},\overline{\gamma}), ~ \forall F \in \mathcal{F} \label{eq:constraint_convergence}\,,%
\end{alignat}
\end{subequations}
where $\mathcal{D}_\theta$ and $\mathcal{D}_0$ represent distributions over the parameters $\theta \in \mathbb{R}^q$ and the initial conditions $\xi_0 \in \Xi_0 \subseteq \mathbb{R}^n$, respectively, $\overline{m}\geq m$ and $\overline{\gamma} \in [\gamma,1)$ in \eqref{eq:constraint_convergence} specify a target  convergence rate over all functions $F \in \mathcal{F}$, and $\mathtt{AlgoCost}(\cdot)$ measures the cost incurred by the algorithm \eqref{eq:algo_def_metaloss_problem} in optimizing a function $F \in \mathcal{F}$ when starting from an initial condition $\xi_0 \in \Xi_0$. We refer to \cite{andrychowicz2016learning,li2017learning} for commonly used algorithm performance metrics.  %

To cast \eqref{eq:metaopt_constrained_problem} as a finite-dimensional learning problem, we perform a number of standard approximations. First, we parametrize classes of distributions $\mathcal{D}_\theta(\psi)$, where $\psi \in \mathbb{R}^{q_\psi}$ is a new set of trainable parameters, that is, we let $\theta \sim \mathcal{D}_\theta(\psi)$. For instance, one can parametrize the mean and variance of Gaussian distributions. Second, we approximate the expectations in \eqref{eq:expected_meta_loss} with their sample-average counterparts
\begin{equation}
    \label{eq:sample_average_approximations}
    \frac{1}{N_\theta N_F N_0} \sum_{i = 1}^{N_\theta} ~ \sum_{j = 1}^{N_F} ~ \sum_{k = 1}^{N_0} ~ \mathtt{AlgoCost}(F_j, \bm{\xi}(\theta_i,{\xi_0}_k))\,,
\end{equation}
on the sets $\hat{\mathcal{F}} = \{F_1, \dots, F_{N_F}\}$, $\hat{\Xi}_{0} = \{\Xi_{0_1}, \dots, \Xi_{0_{N_{0}}}\}$, and $\hat{\Theta} = \{\theta_1, \dots, \theta_{N_\theta}\}$ with elements independently sampled from the distributions $\mathcal{D}_F$, $\mathcal{D}_0$, and $\mathcal{D}_\theta$, respectively. These two approximations allow us to evaluate the empirical algorithm performance explicitly as per \eqref{eq:sample_average_approximations}, and backpropagate through the parameters $\psi$.
\begin{remark}
     The by-design convergence guarantee \eqref{eq:constraint_convergence} ensures that the trajectories $\bm{\xi}$ never diverge, and therefore, commonly chosen algorithmic costs \cite{andrychowicz2016learning} admit a uniform upperbound $\operatorname{AlgoCost}(\cdot)\leq C$ over any bounded sets of parameters $(\theta,F,\xi_0)$, without wrapping the values between intervals defined a-priori \cite{sambharya2024data}. Therefore, the use of stochastic algorithm augmentations enables a direct application of PAC-Bayes generalization bounds \cite{alquier2021user} as follows. Given a prior distribution of augmentation parameters $\theta \sim \mathcal{D}_\theta(\psi_0)$ with compact support,
    and a trained posterior distribution of augmentation parameters $\theta \sim \mathcal{D}_\theta(\psi^\star)$ with compact support,
    one has by \cite[Theorem 2.1]{alquier2021user} that, for any $\lambda > 0$ and $\delta \in (0,1)$,
    \begin{align}
    \label{eq:pac_bayes_bounds}
        \mathbb{P}\Bigg(
            &\mathbb{E}_{\theta \sim \mathcal{D}_\theta(\psi^\star)}[R(\theta)] \leq 
            \mathbb{E}_{\theta \sim \mathcal{D}_\theta(\psi^\star)}[r(\theta)] 
            + \frac{\lambda C^2}{8 \,(N_F + N_0)} \nonumber \\
            &~ + \! \frac{ \mathrm{KL}\!\left(\mathcal{D}_\theta(\psi^\star)\,\middle\|\,\mathcal{D}_\theta(\psi_0)\right) 
            + \log\!\left(\tfrac{1}{\delta}\right)}{\lambda}
        \Bigg) \geq 1-\delta\,,
    \end{align}
    where $R(\theta) \;=\; 
    \mathbb{E}_{\substack{F \sim \mathcal{D}_{\mathcal{F}}\\ \xi_0 \sim \mathcal{D}_{0}}} 
    \big[\, \mathtt{AlgoCost}(F, \bm{\xi}) \,\big],
    $ and $
    r(\theta) \;=\; \frac{1}{N_F N_0} \sum_{j = 1}^{N_F} \sum_{k = 1}^{N_0} 
    \mathtt{AlgoCost}\!\big(F_j, \bm{\xi}(\theta,{\xi_0}_k)\big)$ is its empiric value, and $\mathrm{KL}(\cdot\!\parallel\!\cdot)$ denotes the Kullback--Leibler divergence. The bound holds for every set of samples $\{F_j\}_{j=1}^{N_F}$ and $\{{\xi_0}_k\}_{k=1}^{N_0}$ obtained by sampling from 
    $\mathcal{D}_{\mathcal{F}}$ and $\mathcal{D}_0$, respectively (i.e., for any realization of these draws). The inequality \eqref{eq:pac_bayes_bounds} enables quantifying how much the performance of the algorithm—trained over a set of problems $\hat{\mathcal{F}}$ and $\hat{\Xi}_0$—generalizes to future instances $(F,\xi_0) \sim (\mathcal{D}_\mathcal{F},\mathcal{D}_0)$. We refer the interested reader to more advanced PAC-Bayes bounds \cite{alquier2021user} and their specialization to learned optimization \cite{sambharya2024data} for allowing optimal tuning of the quantities in \eqref{eq:pac_bayes_bounds} as part of the learning procedure.
\end{remark}
Last, we parametrize exponentially decaying augmentations $v_t(F, \xi_{t:0}, \theta)$ such that \eqref{eq:constraint_convergence} holds for any $\theta \in \mathbb{R}^q$. To achieve this, we write $v_t$ as the product of a trainable exponentially decaying sequence of scalars and a NN operator encoding the direction of $v_t$. Specifically, we let $\theta = (\alpha, \varphi)$, where $\alpha \in \mathbb{R}^{\overline{m}+1}_+$ are the coefficients of a polynomial and $\varphi \in \mathbb{R}^{r}$ are the weights of a NN, and let
\begin{align}
\label{eq:MD_form}
    &v_t= \left(\sum_{s=0}^{\overline{m}} \alpha_s t^{s} \right)\overline{\gamma}^t \operatorname{tanh}(D_t(F(\xi_{t:0}),\nabla  F(\xi_{t:0}), \xi_{t:0}, \varphi))\,,
\end{align}
where $D_t$ is chosen as the output of a long short-term memory (LSTM) network with 2 layers. For ease of implementation, in the examples below we consider deterministic algorithm updates by letting $\theta$ be drawn from a Dirac distribution with trainable mean. When dealing with polyhedral constraints as per  Corollary~\ref{co:constrained}, we use Agmon's iterative method \cite{agmon1954relaxation} to enforce \eqref{eq:constraints_v} after generating $v_t$ through \eqref{eq:MD_form}.

\subsection{Augmenting NAG for ill-conditioned systems of linear equations}
In our first example, we consider the problem of solving systems of linear equations of the form $Ax = b$, where $A \succ 0$ and $b$ are sampled from a joint distribution $\mathcal{D}_{A,b}$. Although each instance of this problem admits the analytical solution $x^\star = A^{-1} b$, directly computing the matrix inverse becomes numerically unstable when the condition number $\kappa(A)$ of $A$ grows very large. To address this, we instead approximate $x^\star$ using iterative methods by solving the equivalent quadratic program:
\begin{equation}
    \label{eq:linear_regression_qp}
    \min_{x \in \mathbb{R}^d} ~|A x-b|^2 = \min_{x \in \mathbb{R}^d} ~ x^\top A^\top A x -2b^\top Ax + b^\top b\,.
\end{equation}

For our experiment, we assume that $A = \hat{A} + \delta_A$, where $\hat{A}$ corresponds to the matrix $\mathtt{bcsstk02}$ in the dataset \cite{davis2011university} and is such that $\lambda_{\operatorname{min}}(\hat{A}^\top\hat{A}) \approx 17.75$,  $\kappa(\hat{A}^\top\hat{A}) \approx 1.17\times 10^8$, and $\delta_A$ is a matrix with entries drawn from a standard Gaussian distribution. Similarly, we assume that $b = \hat{b} + \delta_b$, where $\hat{b}$ is a vector with entries $\hat{b}_i = 0.5$ and $\delta_b$ is a vector with entries drawn from a zero-mean Gaussian random variable with standard deviation $0.2$.

We first solve the optimization problem \eqref{eq:linear_regression_qp} using standard gradient descent (GD) and NAG methods, with step-size and momentum chosen according to the optimal tuning for quadratic functions given in \cite[Proposition~1]{lessard2016analysis}. Using the latter method as a baseline optimizer, we then train a two-layer LSTM to learn an augmentation term $v_t$ improving the resulting empirical performance \eqref{eq:sample_average_approximations} over a dataset $\mathcal{F}_{\text{train}}$ of $1024$ realizations of $A$ and $b$. Specifically, we pick the algorithmic cost function in \eqref{eq:metaopt_constrained_problem} as $\mathtt{AlgoCost}(A,b,\mathbf{x})=\sum_{t=0}^T ~|Ax_t - b|^2$ for $T=10000$, and perform meta-optimization  using Adam with a learning rate of $10^{-3}$ for $100$ epochs.

We then construct a test dataset by sampling $256$ independent realizations of $A$ and $b$ and compare the average-case performance of our learned optimizer against standard methods in solving \eqref{fig:linear_regression} in Figure~\ref{fig:linear_regression}. As expected, the introduction of a momentum term enables NAG to converge significantly faster than standard gradient descent. Nevertheless, both methods still require a large number of iterations to reach solutions with high accuracy. Remarkably, we observe that our linearly convergent L2O method learns to initially follow the direction of the \emph{positive} gradient, therefore initially increasing the cost function rather than decreasing it. This behavior does not pertain to classical optimizers and effectively accelerates the accumulation of momentum in the early stages. As demonstrated by Figure~\ref{fig:linear_regression}, this learned behavior results in improved transient performance, without affecting the asymptotic convergence rate.
\begin{figure}[t]
    \centering
    \includegraphics[width=\columnwidth]{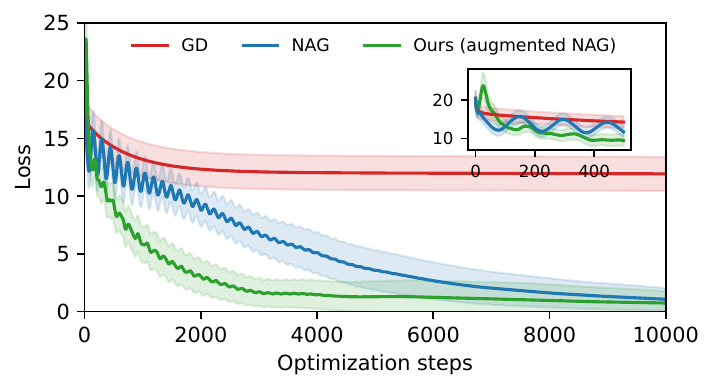}
    \caption{Comparison between the average-case performance of classical and learned optimizers in solving the linear regression problem \eqref{eq:linear_regression_qp}; shaded areas and solid lines denote standard deviations and mean values, respectively.}
    \label{fig:linear_regression}
\end{figure}

Finally, to assess the impact of our linear convergence guarantees, we repeat the training of a 2-layer LSTM optimizer without enforcing an exponential decay of the augmentations $v_t$ in \eqref{eq:optimization_algorithms_fixed_point_plus_v}. We observe that, without clipping the magnitude of $v_t$ as per \eqref{eq:MD_form}, the training procedure becomes highly sensitive to the initial parameter choice $\theta_0$. In fact, since the matrix $A^\top A$ is ill-conditioned, the algorithmic cost function $\mathtt{AlgoCost}(A,b,\mathbf{x})$ dramatically increases if the algorithm updates are poorly chosen, causing numerical instability problems when performing backpropagation to update the parameters $\theta$. As a result, the lack of formal guarantees makes training algorithms to perform a large number of optimization steps numerically challenging. In this example, we thus consider a number of optimization steps of $T = 30$. To make things worse, Figure~\ref{fig:linear_regression_lstm} showcase how learned optimizers may diverge  when unrolled for more optimization steps than they were during training—even when training and test problems are drawn from the same distribution. These observations highlight that restricting our search space to linearly convergent algorithms yield not only theoretical but also practical computational advantages. 

\begin{figure}[t]
    \centering
    \includegraphics[width=\columnwidth]{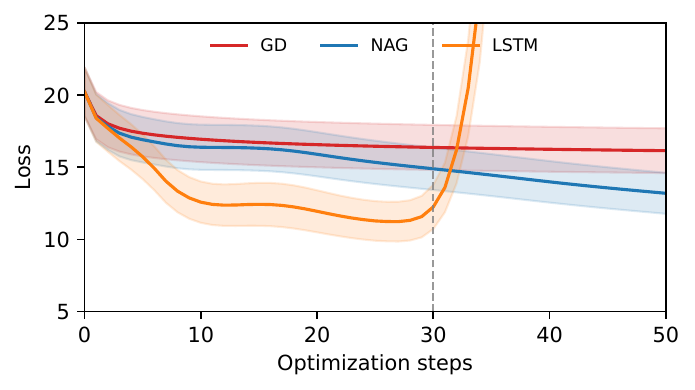}
    \caption{Comparison between the average-case performance of classical and learned optimizers in solving the linear regression problem \eqref{eq:linear_regression_qp}; without convergence guarantees, learned optimizers may diverge when unrolled for more steps than they were during training (vertical dashed line).}
    \label{fig:linear_regression_lstm}
\end{figure}

\subsection{Augmenting projected gradient descent for MPC}
In our second example, we consider a model predictive control setting where only a limited number of optimization steps can be executed in real-time. In particular, we study a discrete-time linear dynamical system described by the state space equation
\begin{equation}
    \label{eq:lti_dynamics_mpc}
    x_{t+1} = A x_t + Bu_t + w_t = \begin{bmatrix}
        1 & 1\\ 0 & 1
    \end{bmatrix} x_t + \begin{bmatrix}
        0\\1
    \end{bmatrix} u_t + w_t\,,
\end{equation}
where $x_t \in \mathbb{R}^2$ is the system state, $u_t \in \mathbb{R}$ is the control input, and $w_t \in \mathbb{R}^2$ represents a zero-mean Gaussian process noise term. 

The goal at each time step is to solve the following finite-horizon linear quadratic control problem:
\begin{subequations}    
\label{eq:mpc_control_problem}
\begin{alignat}{3}
    &~\min_{u_0, \dots, u_{T-1}} ~ &&\sum_{t=0}^{T-1} x_t^\top Q x_t + u_t^\top R u_t + x_T^\top Q_T x_T\\
    &\st ~ &&x_{t+1} = A x_t + Bu_t\,,\\
    & &&u_t \in \mathcal{U}_t\,, x_t \in \mathcal{X}\,, x_0 \in \mathbb{R}^2\,;
\end{alignat}
\end{subequations}
for simplicity, we assume that the weighting matrices $Q, R$ and $Q_T$ are identity matrices of appropriate dimensions. We further set $\mathcal{U}_t = \{u : ||u||_\infty \leq 0.25\}$ and $\mathcal{X}_t = \mathbb{R}^2$ to account for actuation constraints yet sidestep recursive feasibility issues for ease of exposition, as our focus lies in efficiently solving the underlying quadratic program.

By introducing the stacked notation $\mathbf{u} = [u_0^\top \dots u_{T-1}^\top]^\top$ and $\mathbf{x} = [x_0^\top \dots x_T^\top]^\top$, \eqref{eq:mpc_control_problem} can be equivalently rewritten as:
\begin{subequations}
\label{eq:mpc_qp}
\begin{align}
    &~\min_{\mathbf{u}} ~ \mathbf{u}^\top (\mathbf{G}^\top \mathbf{Q} \mathbf{G} + \mathbf{R}) \mathbf{u} + 2 x_0^\top \mathbf{F}^\top \mathbf{Q} \mathbf{G} \mathbf{u} + x_0^\top \mathbf{F}^\top \mathbf{Q} \mathbf{F} x_0 \label{eq:mpc_qp_objective} \\
    &\st ~ \mathbf{x} = \mathbf{F} x_0 + \mathbf{G} \mathbf{u}, \mathbf{u} \in \boldsymbol{\mathcal{U}}, x_0 \in \mathbb{R}^2 \label{eq:mpc_qp_constraints}\,,
\end{align}
\end{subequations}
where $ \mathbf{Q} = \operatorname{blkdiag}(I_T \otimes Q, Q_T)$, $\mathbf{R} = I_T \otimes R$, and $ \mathbf{F} $ and $ \mathbf{G} $ are block matrices comprising the system Markov parameters to encode the system dynamics over a prediction horizon of length $T = 20$.

To solve the quadratic program above, we start from an initial guess $\mathbf{u}^{(0)}$ equal to zero and employ the projected gradient descent (PGD) method \eqref{eq:base_Proj_composite} with step size $\eta = \frac{1}{\lambda_{\operatorname{max}}(\mathbf{G}^\top \mathbf{Q} \mathbf{G} + \mathbf{R})} \approx 3.8 \cdot 10^{-5}$. Using this as our baseline optimization algorithm, we then learn an augmentation term $v_t$ parametrized as per \eqref{eq:MD_form} to minimize the total average cost \eqref{eq:mpc_qp_objective} over $100$ optimization steps $\mathbf{u}^{(1)}, \dots, \mathbf{u}^{(100)}$ when each component of the initial state $x_0$ of the system \eqref{eq:lti_dynamics_mpc} is drawn from a zero-mean Gaussian random variable with standard deviation of $0.5$. Specifically, we perform meta-optimization using Adam with a learning rate of $5 \cdot 10^{-3}$ for $65$ epochs. 
\begin{figure}[t]
    \centering
    \includegraphics[width=\columnwidth]{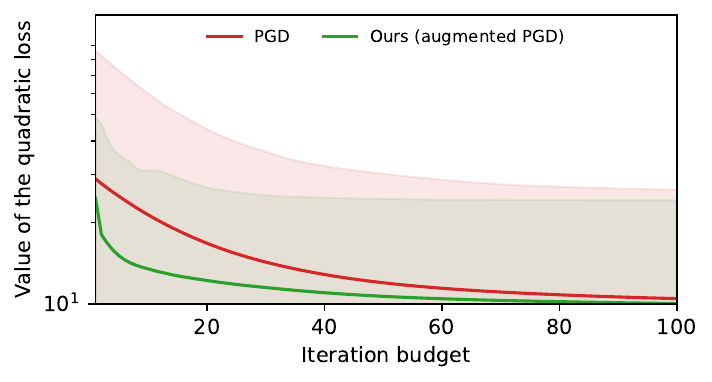}
    \caption{Comparison between the average-case performance of classical and learned optimizers in solving the quadratic program \eqref{eq:mpc_qp}; solid lines represent mean values, and shaded areas cover values up to the $90\%$ percentile.}
    \label{fig:random_qp_x0_gaussian}
\end{figure}
As highlighted by Figure~\ref{fig:random_qp_x0_gaussian}, while both algorithms converge in a small neighborhood of the solution of \eqref{eq:mpc_qp} after $100$ iterations, our learned optimizer showcase improved transient performance. Numerically, we verify that this is because during training $v_t$ learns to promptly saturate the constraint $u_t \in \mathcal{U}_t$ when needed, leading to better performance in the testing phase. Remarkably, as shown in Figure~\ref{fig:random_qp_closed_loop_mpc}, we observe a similar trend even when we close the loop between the linear dynamical system \eqref{eq:lti_dynamics_mpc} and our learned optimizer. In particular, despite the repeated solution of \eqref{eq:mpc_qp} in a receding horizon fashion induces a new distribution over the state space, using our learned optimizer—trained on initial conditions $x_0$ randomly drawn from a Gaussian distribution—enables a significant reduction of the closed-loop cost incurred over an horizon of length $N =30$ when only a limited optimization steps can be performed in real-time.
\begin{figure}[t]
    \centering
    \includegraphics[width=\columnwidth]{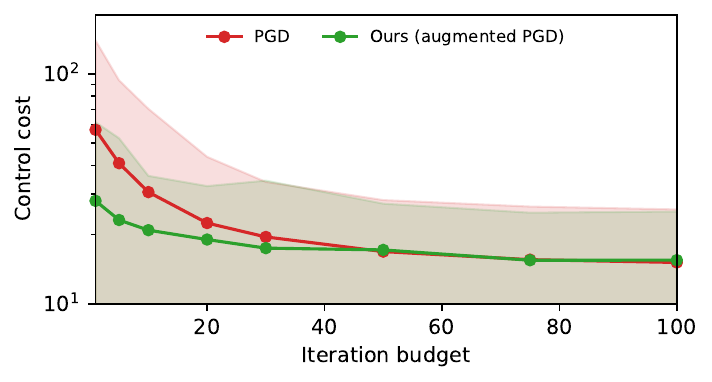}
    \caption{Comparison between the closed-loop cost incurred by the receding horizon control law that results from unrolling classical and learned optimizer to approximate the solution of the quadratic program \eqref{eq:mpc_qp}; solid lines represent mean values, and shaded areas cover values up to the $90\%$ percentile.}
    \label{fig:random_qp_closed_loop_mpc}
\end{figure}

\section{Conclusion}
We have shown that many classical linearly convergent algorithms—ranging from standard gradient descent to accelerated and proximal methods—can be enriched with an exponentially-decaying augmentation term whose  role is to improve average‐case performance on a specified distribution of problems. Crucially, this can be achieved without sacrificing their provable worst‐case linear rates. Specifically, (i) these perturbations can be injected at any desired frequency below a threshold, trading off average-case improvement against worst-case rate degradation in a quantifiable way; and (ii) every regular linearly convergent method admits this form of decomposition.  In practice, these results enable average-case improvement of baseline algorithms for composite optimization, by designing an exponentially decaying update function akin to designing a robustly stabilizing feedback policy for nonlinear control systems. Numerical examples confirm  the potential for significant reduction in the number of optimization steps over general-purpose solvers.

While we have focused on application to learned optimization and average-case improvement, the characterization of all linearly convergent algorithms has independent interest. Accordingly, one direction of interest is to exploit the characterization to deriving update rules with optimal performance from the lens of typical control costs, drawing novel parallels between optimal control theory and accelerated optimization. Important venues of further investigation include analyzing the computational overhead arising from computing learned augmentations, guarantees for derivative-free optimization, design of nonexpansive fixed-point iterations, and applications to time-varying and online optimization.

\bibliographystyle{plain}        %
\bibliography{references}           %

\appendix
\appendix

\section{Proof of Theorem~\ref{th:sufficiency_polyexp_general}}
\label{app:proof_sufficiency_polyexp_general}
We first prove the result by assuming that $\pi \in \Exp{\gamma}$. This is instrumental towards establishing the general result. Let $\delta_t = \operatorname{dist}(\xi_t, \Fix{\pi})$ for compactness. By the algorithm definition \eqref{eq:optimization_algorithms_fixed_point_plus_v} and the triangle inequality, we have that for every $F \in \mathcal{F}$ 
\begin{align*}
    \delta_{t} &= \operatorname{dist}(\pi(F,\xi_{t-1}) + v_{t-1}, \Fix{\pi})\\
    &= \inf_{c \in \Fix{\pi}} ~ \operatorname{dist}(\pi(F,\xi_{t-1}) + v_{t-1}, c)\\
    &\leq \inf_{c \in \Fix{\pi}} ~ \operatorname{dist}(\pi(F,\xi_{t-1}), c) + |v_{t-1}|\\
    &= \operatorname{dist}(\pi(F,\xi_{t-1}), \Fix{\pi}) + |v_{t-1}|\,.
\end{align*}
Assuming that $\pi \in \Exp{\gamma}$, we have that \eqref{eq:exponential_convergence} holds with $p(t) = 1$. It follows that $\delta_{t} \leq \gamma \delta_{t-1} + |v_{t-1}|$. Iterating this inequality, we deduce that
\begin{align*}
    \delta_{t} &\leq \gamma^{t} \delta_0 + \sum_{k=0}^{t-1} \gamma^k |v_{t-1-k}|\\
    &\leq \gamma^{t} \delta_0 + \sum_{k=0}^{t-1} \gamma^k p(t-1-k) \gamma^{t-1-k}\\
    &\leq \gamma^{t} \left(\delta_0 + \frac{1}{\gamma} \sum_{k=0}^{t-1} p(k)\right)\,,
\end{align*}
where we used the fact that $\mathbf{v} \in \ell_{exp}(m, \gamma)$. Let $q(t) = \sum_{k=0}^t p(k)$ and note that the right-hand side of the above can be written as $\gamma^t r(t)$ where $r(t) = \delta_0 + \frac{1}{\gamma} q(t-1)$. We study $q(t)$. By linearity of summation, $q(t)$ can be equivalently rewritten as
\begin{equation*}
    \sum_{k=0}^{t} \sum_{j = 0}^m a_j k^j = a_m \sum_{k = 0}^{t} k^m + \dots + a_1 \sum_{k = 0}^{t} k + a_0 \sum_{k = 0}^{t} 1\,,
\end{equation*}
where $a_j \in \mathbb{R}$ with $j \in \{0, \dots, m\}$ is the $j$-th coefficient of the polynomial $p(\cdot)$. Faulhaber's formula implies that $q(t)$ is a polynomial of degree $m+1$ in the variable $t$ with coefficient $b_{m+1} = \frac{a_m}{m+1}$. Furthermore, $q(t)$ is positive and monotonically non-decreasing by construction, that is, $q(t) \in \mathcal{P}_{m+1}(t)$. Note that $q(t) \in \mathcal{P}_{m+1}(t)$ implies $r(t) \in \mathcal{P}_{m+1}(t)$. Hence, we conclude that $\delta_{t} \leq r(t) \gamma^{t}$ for all $t \in \mathbb{N}$, which proves the result for the case $\pi \in \Exp{\gamma}$.

We now turn our attention to the general case where $\pi$ is any linearly convergent algorithm in  $\pExp{m,\gamma}$ certified by a polynomial $p(t) \in \mathcal{P}_m(t)$. We distinguish two cases. In the case $p(1) \gamma < 1$, we have that $\pi \in \pExp{m,\gamma}$ also lies in $\Exp{p(1) \gamma}$, and the analysis above applies with rate $p(1) \gamma$ in place of $\gamma$. 

In the case $p(1) \gamma \geq 1$, we necessarily have than $N \geq 2$. For any $w_t \in \ell_{exp}(m,\gamma)$, consider the recursion
\begin{equation}
    \label{eq:optimization_algorithms_fixed_point_plus_v_stacked}
    \zeta_{k+1} = \pi^{N}(F, \zeta_k) + w_t\,,
\end{equation}
where $\pi^N(\cdot)$ denotes the repeated application of $\pi$ defined as $\pi^{N}(F, \zeta_k) = \pi(F, \pi^{N-1}(F, \zeta_k))$ with $\pi^1(F, \zeta_k) = \pi(F, \zeta_k)$. 

We first observe that, if $\zeta_0 = \xi_0$ and $\mathbf{v}$ is constructed as per \eqref{eq:v_construction_zeros}, then \eqref{eq:optimization_algorithms_fixed_point_plus_v_stacked} is equivalent to  \eqref{eq:optimization_algorithms_fixed_point_plus_v} in the sense that  $\zeta_{k} = \xi_{Nk}$ for every $k \in \mathbb{N}$. By construction, $\pi^{N} \in \Exp{\rho}$ and therefore complies with \eqref{eq:exponential_convergence} with $p(t) =1$ and $\gamma=\rho$. Hence, as proven above, it holds that
\begin{equation*}
    \operatorname{dist}(\bm{\zeta}, \Fix{\pi^N}) \in \ell_{exp}(m+1, \rho)\,.
\end{equation*}
We now argue that $\Fix{\pi^N} = \Fix{\pi}$. Clearly, $\Fix{\pi^N} \supseteq \Fix{\pi}$ since $\xi^{\star} = \pi(F, \xi^{\star})$ for every $\xi^{\star} \in \Fix{\pi}$ and thus $\xi^{\star} = \pi^N(F, \xi^{\star})$. To show that $\Fix{\pi^N} \subseteq \Fix{\pi}$, assume there exists $\zeta^\star \in \Fix{\pi^N}$ such that $\zeta^\star \notin \Fix{\pi}$. Since $\pi \in \pExp{m,\gamma}$, we have that $\lim_{t \to \infty} ~ \operatorname{dist}(\pi^t(F,\zeta^\star), \Fix{\pi}) = 0$. At the same time, $\operatorname{dist}(\pi^{\tau N}(F,\zeta^\star), \Fix{\pi}) > 0$ for any $\tau \in \mathbb{N}$ because $\pi^{\tau N}(F,\zeta^\star) = \zeta^\star \notin \Fix{\pi}$. This is a contradiction, and thus $\Fix{\pi^N} = \Fix{\pi}$. We conclude that 
\begin{equation*}
    \operatorname{dist}(\bm{\zeta}, \Fix{\pi}) \in \ell_{exp}(m+1, \rho)\,,
\end{equation*}
and therefore there exists a polynomial $q(k) \in \mathcal{P}_{m+1}(k)$ such that $\operatorname{dist}(\zeta_k, \Fix{\pi}) \leq q(k) \rho^k$ for all $k \in \mathbb{N}$. 

Next, we note that, for any $s \in \{1,\dots, N-1\}$
\begin{align*}
    \operatorname{dist}(\xi_{Nk + s}, \Fix{\pi}) &= \operatorname{dist}(\pi^s(F,\xi_{Nk}), \Fix{\pi})\\ &\leq p(s) \gamma^s \operatorname{dist}(\xi_{Nk}, \Fix{\pi})\\
    &\leq p(s) \gamma^s q(k) \rho^k\,,
\end{align*}
where we used the fact that $\operatorname{dist}(\xi_{Nk}, \Fix{\pi}) = \operatorname{dist}(\zeta_k, \Fix{\pi})$ for any $k \in \mathbb{N}$. Letting $t = Nk + s$, and using the fact that $p(\cdot), q(\cdot) \in \mathcal{P}_{m+1}$, we obtain
\begin{align*}
    \operatorname{dist}(\xi_t, \Fix{\pi}) &\leq p(N-1) \gamma q\left(\left \lfloor \frac{t-s}{N} \right \rfloor\right) \rho^{\left \lfloor \frac{t-s}{N} \right \rfloor}\\
    &\leq p(N-1) \gamma q\left(\frac{t}{N}\right) \rho^{\left \lfloor \frac{t-N+1}{N} \right \rfloor}\\
    &\leq \underbrace{\frac{p(N-1) \gamma}{\rho^{2-\frac{1}{N}}} q\left(\frac{t}{N}\right)}_{r(t) \in \mathcal{P}_{m+1}(t)} \left(\rho^{\frac{1}{N}}\right)^{t}\,.
\end{align*}
Since $\rho = p(N)\gamma^N$, we have that $\rho^{\frac{1}{N}}= \sqrt[N]{p(N)} \gamma$. This concludes the proof.

\section{Proof of Theorem~\ref{th:necessity}}
Let $v_t(F, \xi_{t:0}) = -\pi(F, \chi_t) + \sigma_t(F, \chi_{t:0})$. We first show by induction that $\xi_t = \chi_t$ at all times, starting from the base case $\xi_0 = \chi_0$, which holds by construction. Assume now that $\xi_{t:0} = \chi_{t:0}$. We aim to prove that $\xi_{t+1} = \chi_{t+1}$. This holds because
\begin{align*}
    \xi_{t+1} &= \pi(F, \xi_t) - \pi(F, \chi_t) + \sigma_t(F, \chi_{t:0})\\
    &= \sigma_t(F, \chi_{t:0}) = \chi_{t+1}\,.
\end{align*}
It remains to show that the sequence $v_t(F, \xi_{t:0}) = -\pi(F, \chi_t) + \sigma_t(\chi_{t:0})$ belongs to $\ell_{exp}(m, \gamma)$. To prove this, we rewrite $v_t$ as
\begin{align}
\label{eq:matching_v}
    v_t = - (\pi(F, \chi_t) - \chi_t) + \sigma_t(F, \chi_{t:0}) - \chi_t\,.
\end{align}
Since $\pi(F,\cdot)$ is Lipschitz continuous, letting $\chi_t^p$ be any element of $\argmin_{\chi \in \Fix{\pi}} |\chi - \chi_t|^2$, there exists a constant $L_\pi \in \mathbb{R}_+$ such that
\begin{align*}
    |\pi(F,\chi_t)-\chi_t|&=|\pi(F,\chi_t)- \chi_t^p+\chi_t^p-\chi_t|\\
    &=  |\pi(F,\chi_t)- \pi(F,\chi_t^p)+\chi_t^p-\chi_t|\\
    &\leq (L_\pi+1)|\chi_t- \chi_t^p|\\
    &= (L_\pi+1)\operatorname{dist}(\chi_t,\Fix{\pi})\,,
\end{align*}
and hence  $- (\pi(F, \bm{\chi}) - \bm{\chi}) \in \ell_{exp}(m,\gamma)$. We further have that $\bm{\sigma}(F,\bm{\chi})-\bm{\chi} \in \ell_{exp}(m,\gamma)$ by the regularity assumption on $\bm{\sigma}$ as per Definition~\ref{def:regularity}. Since the sum of signals in $\ell_{exp}(m, \gamma)$ belongs to $\ell_{exp}(m, \gamma)$, we conclude the proof by inspection of \eqref{eq:matching_v}.

\end{document}